\documentclass[11pt,a4paper,fleqn,leqno]{article}
\usepackage[utf8]{inputenc}
\usepackage{amsmath}
\usepackage{amsfonts}
\usepackage{amssymb}
\usepackage{lmodern}
\usepackage{amsthm}
\usepackage{authblk}
\newcommand{\sneg}{\mathord{\sim}}
\newcommand{\B}{\Box}
\newcommand{\D}{\Diamond}

\theoremstyle{plain}
\newtheorem{definition}{Definition}[section]
\newtheorem{theorem}[definition]{Theorem}
\newtheorem{lemma}[definition]{Lemma}
\newtheorem{proposition}[definition]{Proposition}

\theoremstyle{definition}

\newtheorem{example}[definition]{Example}

\title{\bf Propositional dynamic logic with Belnapian truth values}
\author[1]{Igor Sedl\'{a}r\thanks{E-mail: sedlar@cs.cas.cz. This work was supported by the long-term strategic development financing of the Institute of Computer Science (RVO:67985807) and the Czech Scientific Agency grant GBP202/12/G061. Previous versions of this article were presented at the Knowledge Representation Seminar of the Department of Applied Informatics, Faculty of Mathematics, Physics and Informatics, Comenius University in Bratislava, and the Applied Mathematical Logic Seminar of the Department of Theoretical Computer Science, Institute of Computer Science, The Czech Academy of Sciences. I am grateful to the audiences at these seminars for useful feedback. Martin Bal\'{a}\v{z}, Petr Cintula, Rostislav Hor\v{c}\'{i}k and J\'{a}n \v{S}efr\'{a}nek have provided especially helpful comments. Thanks are also due to the AiML reviewers for suggesting a number of improvements.}}
\affil[1]{Institute of Computer Science, The Czech Academy of Sciences, Pod Vod\'{a}renskou v\v{e}\v{z}\'{i} 271/2, 182 07 Prague 8, Czech Republic}
\date{ }

\begin{document}
\maketitle

\begin{small}

\paragraph{Abstract.} We introduce $\mathbf{BPDL}$, a combination of propositional dynamic logic $\mathbf{PDL}$ with the basic four-valued modal logic $\mathbf{BK}$ studied by Odintsov and Wansing (`Modal logics with Belnapian truth values', J. Appl. Non-Class. Log. 20, 279--301 (2010)). We modify the standard arguments based on canonical models and filtration to suit the four-valued context and prove weak completeness and decidability of $\mathbf{BPDL}$.

\end{small}

\section{Introduction}
Propositional dynamic logic $\mathbf{PDL}$ is a well-known logical framework that allows to express properties of regular programs and formalises reasoning about these properties \cite{FL1979,Harel2000}. The framework sees programs as state transitions, or binary relations on states, where states of the computer are viewed as complete and consistent possible worlds. A more general notion of computer state has been put forward by Belnap and Dunn \cite{Belnap1977a,Belnap1977b,Dunn1976}. In a possible world, every formula is either true or false. In a Belnap--Dunn state, formulas can be (only) true, (only) false, both true and false, or neither true nor false. Informally, Belnap--Dunn states are seen as bodies of information about some domain and the four truth values correspond to presence or absence of information about the domain. More precisely, the four possible truth values of a formula $\phi$ express four possible answers to the query `What is the available information about $\phi$?', namely: 
\begin{itemize}
\item there is information that $\phi$ is true and no information that $\phi$ is false (`true');
\item there is information that $\phi$ is false and no information that $\phi$ is true (`false');
\item there is information that $\phi$ is true but also information that $\phi$ is false (`both');
\item there is no information about $\phi$ (`neither').
\end{itemize}   
\noindent Belnap and Dunn stress the importance of this generalisation to computer science, pointing mainly to databases as a potential area of application. Later work on bilattices, a generalisation of the Belnap--Dunn notion of state, has confirmed their assessment and extended the applications to other areas \cite{Ginsberg1988,Gargov1999,Arieli1996,Fitting1991b,Fitting2006}. 

Putting things together, a version of $\mathbf{PDL}$ using Belnap--Dunn states would formalise reasoning about regular programs that modify (possibly incomplete and inconsistent) database-like structures. Such structures abound and a logical formalisation of reasoning about their algorithmic transformations could be of vital importance to AI and related areas. In addition to practical applications, theoretical questions pertaining to the properties of such generalised versions of $\mathbf{PDL}$ are interesting in their own right. However, Belnap--Dunn versions of $\mathbf{PDL}$ are yet to be investigated.

This article fills the gap. We discuss $\mathbf{BPDL}$, a logic that adds program modalities to Odintsov and Wansing's \cite{Odintsov2010} basic modal logic with Belnapian truth values $\mathbf{BK}$ (see also \cite{Odintsov2012,Odintsov2016}). Our main technical results concerning $\mathbf{BPDL}$ (introduced in Section \ref{sec: BPDL} of the article) are a decidability proof using a variation of the standard argument based on filtration (Section \ref{sec: decidability}) and a sound and weakly complete axiomatisation (Section \ref{sec: completeness}). We assume familiarity with $\mathbf{PDL}$, but a short overview of $\mathbf{BK}$ is provided in Section \ref{sec: BK}.

We note that there are other well-known four-valued modal logics, but there are reasons to favour $\mathbf{BK}$ when it comes to combinations with $\mathbf{PDL}$. Priest's basic modal First-Degree-Entailment $\mathbf{K}_\mathbf{FDE}$ \cite{Priest2008} lacks a sensible implication connective (e.g., Modus ponens fails), which is a problem given the importance of implication in stating properties of programs such as partial correctness. Goble's $\mathbf{KN4}$ \cite{Goble2006} corresponds to a fragment of $\mathbf{BK}$. The framework of Rivieccio, Jung and Jansana \cite{Rivieccio2015} is more complicated than $\mathbf{BK}$ in that it treats the modal accessibility relation itself as many-valued. As a result, for instance, the familiar `$\mathbf{K}$ axiom' $ \B(\phi \to \psi) \to (\B\phi \to \B\psi)$ is not valid. This is problematic from the viewpoint of $\mathbf{PDL}$ which is a normal modal logic. (However, a non-normal version of $\mathbf{PDL}$ built on this framework might still be interesting to look at in the future.) Another approach is to add to $\mathbf{PDL}$ a modal DeMorgan negation in the style of \cite{FHV95}. However, the modal negation in this framework does not fit in with implication as nicely as the negation in $\mathbf{BK}$ (for instance, $\sneg(\phi\to\psi)$ does not entail $ \phi$, where `$\sneg$' is the DeMorgan negation). Nevertheless, this approach is pursued by the present author in \cite{Sedlar2016}.

The general idea of providing many-valued versions of $\mathbf{PDL}$ is not new. Teheux \cite{Teheux2014} formulates $\mathbf{PDL}$ over finitely-valued {\L}ukasiewicz logics to model the R\'{e}nyi--Ulam searching game with errors. However, the non-modal fragments of his logics are non-classical, as opposed to $\mathbf{BK}$ which can be seen as an extension of the classically-based logic $\mathbf{K}$ with a strong negation. B\v{e}hounek \cite{behounek08,bbc08} suggests that $\mathbf{PDL}$ with fuzzy accessibility relations is suitable for reasoning about costs of program executions, but the states in his models remain classical.

\section{Modal logic with Belnapian truth values}\label{sec: BK}
This section provides background on $\mathbf{BK}$ and motivates our extension of the logic with program modalities. The language $\mathcal{L}_{mod}$ consists of $AF$, a countable set of atomic formulas, a nullary connective $\bot$, unary connectives $\sneg, \B, \D$ and binary connectives $\land, \lor, \to$. $\neg \phi$ is defined as $\phi \to \bot$, $\top$ is defined as $\neg \bot$ and $\phi  \leftrightarrow \psi$ is defined as $(\phi \to \psi) \land (\psi \to \phi)$. $F_{mod}$ is the set of formulas of $\mathcal{L}_{mod}$.

\begin{definition}{\cite[285--286]{Odintsov2010}}
An \emph{Odintsov--Wansing model} is a tuple $M = \langle S, R, V^{+}, V^{-} \rangle$ where $S \neq \emptyset$, $R \subseteq (S \times S)$ and $V^{\circ}: AF \mapsto 2^{S}$, $\circ = \{+,-\}$. Every $M$ induces a pair of relations $\models_M^{+}, \models_M^{-} \: \subseteq\: (S \times F_{mod})$ such that (we usually drop the subscript `$M$'):
\begin{enumerate}
\item $x \models^{+} p$ iff $x \in V^{+}(p)$; $x \models^{-} p$ iff $x \in V^{-}(p)$
\item $x \models^{+} \bot$ for no $x$; $x \models^{-} \bot$ for all $x$
\item $x \models^{+} \sneg \phi$ iff $x \models^{-} \phi$; $x \models^{-} \sneg \phi$ iff $x \models^{+} \phi$
\item $x \models^{+} \phi \land \psi$ iff $x \models^{+} \phi$ and $x \models^{+} \psi$; $x \models^{-} \phi \land \psi$ iff $x \models^{-} \phi$ or $x \models^{-} \psi$
\item $x \models^{+} \phi \lor \psi$ iff $x \models^{+} \phi$ or $x \models^{+} \psi$; $x \models^{-} \phi \lor \psi$ iff $x \models^{-} \phi$ and $x \models^{-} \psi$
\item $x \models^{+} \phi \to \psi$ iff $x \not\models^{+} \phi$ or $x \models^{+} \psi$; $x \models^{-} \phi \to \psi$ iff $x \models^{+} \phi$ and $x \models^{-} \psi$
\item $x \models^{+} \B \phi$ iff for all $y$, if $Rxy$, then $y \models^{+} \phi$\\ $x \models^{-} \B \phi$ iff there is $y$ such that $Rxy$ and $y \models^{-} \phi$
\item $x \models^{+} \D \phi$ iff there is $y$ such that $Rxy$ and $y \models^{+} \phi$\\ $x \models^{-} \D \phi$ iff for all $y$, if $Rxy$, then $y \models^{-} \phi$
\end{enumerate} 	
$|\phi|^{+}_M = \{ x \mid x \models^{+}_M \phi \}$ and $ |\phi|^{-}_M = \{x \mid x \models^{-}_M \phi\}$. Entailment in the resulting logic, $\mathbf{BK}$, is defined as $\models^{+}$-preservation in every state of every model ($X \models_\mathbf{BK} \phi$ iff, for all $M$, $\bigcap_{\psi \in X}|\psi|^{+}_M \subseteq |\phi|^{+}_M$). Validity is defined as usual ($\phi$ is valid in $\mathbf{BK}$ iff $\emptyset \models_\mathbf{BK} \phi$).
\end{definition}

States $x \in S$ can be seen as database-like bodies of information. The fact that $x \models^{+} \phi$ can then be read as `$x$ provides information that $\phi$ is true' (or `$x$ supports $\phi$', `$x$ verifies $\phi$') and $x \models^{-} \phi$ as `$x$ provides information that $\phi$ is false' (`$x$ falsifies $\phi$'). Consequently, $|\phi|^{+}$ is seen as the set of states in which $\phi$ is true (the truth set of $\phi$) and $|\phi|^{-}$ as the set of states in which $\phi$ is false (falsity set). Entailment then boils down to the usual notion of truth-preservation. The distinguishing feature of the Belnap--Dunn picture is that some bodies of information $x$ may support conflicting information about some $\phi$ (if $x \models^{+} \phi$ and $x \models^{-} \phi$) and some bodies of information $x$ may not provide any information about some $\phi$ at all (if $x \not\models^{+} \phi$ and $x \not\models^{-} \phi$). In other words, $|\phi|^{+}$ and $|\phi|^{-}$ may have a non-empty intersection and their union is not necessarily identical to $S$. 

The two negations `$\sneg$' and `$\neg$' can be explained as follows. The formula $\sneg \phi$ may be read as `$\phi$ is false' (as $x \models^{+} \sneg \phi$ iff $x \models^{-} \phi$). On the other hand, the formula $\neg \phi$ is read as `$\phi$ is not true' (note that $x \models^{+} \phi \to \bot$ iff $x \not\models^{+} \phi$). In general, neither $\sneg \phi \to \neg \phi$ nor $\neg \phi \to \sneg \phi$ are valid. In other words, the present framework treats `false' and `not true' as two independent notions. The presence of `$\sneg$' and `$\neg$' in our language allows to express the four possible Belnapian truth values of a formula $\phi$:% by formulas containing $\phi$ as a subformula, namely:
\begin{itemize}
\item $\phi \land \neg\sneg \phi$ ($\phi$ is only true, i.e., true and not false); 
\item $\neg \phi \land \sneg \phi$ ($\phi$ is only false, i.e., not true and false);
\item $\phi \land \sneg \phi$ ($\phi$ is both true and false);
\item $\neg \phi \land \neg\sneg \phi$ ($\phi$ is neither true nor false).
\end{itemize}

\begin{theorem}\label{thm: bk complete}
The following axiom system, $H(\mathbf{BK})$, is a sound and strongly complete axiomatisation of $\mathbf{BK}$:
\begin{enumerate}
\item Axioms of classical propositional logic in the language $\{ AF, \bot, \to, \land, \lor \}$ and Modus ponens;
\item Strong negation axioms:
	\begin{align*}
	\sneg \sneg \phi  & \leftrightarrow \phi,\\
	\sneg (\phi \land \psi) & \leftrightarrow (\sneg \phi \lor \sneg \psi),\\
	\sneg (\phi \lor \psi) & \leftrightarrow (\sneg \phi \land \sneg \psi),\\
	\sneg (\phi \to \psi) & \leftrightarrow (\phi \land \sneg \psi),\\
	\top &\leftrightarrow \sneg \bot;
	\end{align*}
\item The $\mathbf{K}$ axiom $\B (\phi \to \psi) \to (\B \phi \to \B \psi)$ and the Necessitation rule $\phi / \B \phi$;
\item Modal interaction principles:
	\begin{align*}
	\neg \B \phi & \leftrightarrow \D \neg \phi,\\
	\neg \D \phi & \leftrightarrow \B \neg \phi,\\
	\sneg \B \phi & \leftrightarrow \D \sneg \phi,\\
	\B \phi & \leftrightarrow \sneg \D \sneg \phi,\\
	\sneg \D \phi & \leftrightarrow \B \sneg \phi,\\
	\D \phi & \leftrightarrow \sneg \B \sneg \phi.
	\end{align*}
\end{enumerate}
\end{theorem}
\begin{proof}
See \cite{Odintsov2010}.
\end{proof}

The logic $\mathbf{BK}$ enjoys the deduction theorem in the sense that $\phi \models \psi$ iff $\models \phi \to \psi$.\footnote{Proof: $\not\models \phi \to \psi$ iff, for some $x$, $x \not\models^{+} \phi \to \psi$ iff, for some $x$, $x \models^{+} \phi$ and $x \not\models^{+} \psi$ iff $\phi \not\models \psi$.} An interesting feature of $\mathbf{BK}$ is that the set of valid formulas is not closed under the Replacement rule $\phi \leftrightarrow \psi / \chi(\phi) \leftrightarrow \chi(\psi)$.\footnote{Note, for example, that $\sneg (\phi \to \psi) \leftrightarrow (\phi \land \sneg \psi)$ is valid by the completeness theorem but $\sneg \sneg (\phi \to \psi) \leftrightarrow \sneg (\phi \land \sneg \psi)$ is not. The latter is provably equivalent to $(\phi \to \psi) \leftrightarrow (\sneg \phi \lor \psi)$. Now consider a model where $x \not\models^{+} p$, $x \not\models^{-} p$ and $x \not\models^{+} q$. Then $x \models^{+} p \to q$ but $x \not\models^{+} \sneg p \lor q$. By the deduction theorem, $(p \to q) \to (\sneg p \lor q)$ is not valid. It is easily shown that the converse implication is not valid either.} However, it is closed under the Positive replacement rule $\phi \leftrightarrow \psi / \gamma (\phi) \leftrightarrow \gamma (\psi)$ for $\sneg$-free $\gamma$ and the Weak replacement rule $(\phi \leftrightarrow \psi) \land (\sneg \phi \leftrightarrow \sneg \psi) / \chi(\phi) \leftrightarrow \chi(\psi)$. (See \cite{Odintsov2010} for details.) Schemas $(\phi \land \sneg \phi) \to \bot$ and $\phi \lor \sneg \phi$ are not valid (but, of course, $(\phi \land \neg \phi) \to \bot$ and $\phi \lor \neg \phi$ both are). 

Languages interpreted over bilattices often contain two additional binary connectives `$\otimes$' and `$\oplus$'. Their meaning can be outlined by the following example using the reading of the four Belnapian truth values as subsets of the set of `classical' values $\{\text{true}, \text{false}\}$. If $\phi$ is only true and $\psi$ is only false (the value of $\phi$ is $\{\text{true}\}$ and the value of $\psi$ is $\{\text{false}\}$), then $\phi \otimes \psi$ is neither true nor false ($\{\text{true}\} \cap \{\text{false}\} = \emptyset$) whereas $\phi \oplus \psi$ is both true and false ($\{\text{true}\} \cup \{\text{false}\} = \{ \text{true}, \text{false}\}$).\footnote{Closer  to the present setting,  $\phi \otimes \psi$ is taken to be verified (falsified) iff both $\phi$ and $\psi$ are verified (falsified); and $\phi \oplus \psi$ is verified (falsified) iff at least one of $\phi, \psi$ is verified (falsified). (Hence, for example, extending $\mathbf{BK}$ with these connectives would result in $(\phi \otimes \psi) \leftrightarrow (\phi \land \psi)$ and $(\sneg \phi \otimes \sneg \psi) \leftrightarrow (\sneg \phi \land \sneg \psi)$ being both valid, and similarly for $\oplus$ and $\lor$.)} Odintsov and Wansing \cite{Odintsov2010} do not use these connectives in the modal setting and, for the sake of simplicity, we omit them as well. We note, however, that there is no technical obstacle in introducing them to the framework and, speaking in terms of informal interpretation, they fit in nicely also to our combination of $\mathbf{BK}$ with $\mathbf{PDL}$.

Let us now return to the informal interpretation of $\mathbf{BK}$. If states in the model are seen as database-like bodies of information, then the accessibility relation can be construed as any binary relation between such bodies of information. Interpretations related to transformations of such bodies (adding or removing information, for example) are a natural choice. For instance, with a set of available transformations in mind, we may read $Rxy$ as `$y$ is the result of transforming $x$ in some available way'. $\D \phi$ then means that there is an available transformation of the present body of information that leads to $\phi$ being supported and $\B \phi$ means that all available transformations lead to $\phi$ being supported. Hence, $\mathbf{BK}$ can be seen as a general formalism for reasoning about such transformations.

This reading of $R$ invites us to generalise the framework to a multi-modal setting. We may want to distinguish between different types of transformation and so we may need $R_i$ for each type $i$ instead of a single relation $R$. The corresponding formulas of a multi-modal extension of $\mathcal{L}_{mod}$, $\B_i\phi$ ($\D_i\phi$), would then express that $\phi$ is supported after every (some) transformation of type $i$. With a number of basic types at hand, the natural next step is to introduce complex transformations consisting of transformations of the basic types. This brings us to extending $\mathbf{BK}$ with program operators provided by $\mathbf{PDL}$, i.e., choice, composition, iteration and test. Additional motivation for considering a combination of $\mathbf{PDL}$ with $\mathbf{BK}$ is given by the following examples.

\begin{example}\label{exam: default}
If Belnap--Dunn states are seen as bodies of information, then state transitions (programs) may be seen as general inference rules. Formulas of the combined language may express the nature and properties of these rules. Introducing a Belnapian negation $\sneg$ into the language of $\mathbf{PDL}$ opens the possibility of expressing inferences beyond the scope of classical logic. Take, for example, \emph{default rules} of the form
\begin{equation}\label{eq: default rule}
\frac{\psi : \phi}{\chi}
\end{equation}
read `If $\psi$ is true and there is no information that $\phi$ is false, then infer that $\chi$ is true'. Such a default rule may  be expressed by
\begin{equation}\label{eq: default formula}
(\psi \land \neg\sneg \phi) \to [\alpha]\chi,
\end{equation}
a formula that reads `If $\psi$ is true and $\phi$ is not false, then every terminating execution of $\alpha$ leads to a state where $\chi$ true'. If \eqref{eq: default formula} holds in a state, then executing the program $\alpha$ in the state is equivalent to using \eqref{eq: default rule} in the state. Hence, \eqref{eq: default rule} and $\alpha$ are `locally equivalent' in the given state. Moreover, if
\begin{equation*}
[\beta^{\ast}]\left((\psi \land \neg\sneg \phi) \to [\alpha]\chi\right)
\end{equation*}
holds in a state, then \eqref{eq: default rule} and $\alpha$ are `$\beta$-equivalent', or locally equivalent in every state reachable by a finite iteration of $\beta$. 

Formulas of the form \eqref{eq: default formula} may even be seen as \emph{defining} $\alpha$ to be a counterpart of a specific default rule. On this view, it is natural to focus only on models where \eqref{eq: default formula} holds in every state (is valid). This motivates a notion of global consequence to be introduced below.    
\end{example}

\begin{example}
A special case of \eqref{eq: default rule} is the \emph{closed-world assumption} rule
\begin{equation}\label{eq: cwa rule}
\frac{\top: \neg \phi}{\sneg \phi},
\end{equation}
inferring that $\phi$ is false from the assumption that $\phi$ is not known to be true. Applications of \eqref{eq: cwa rule} correspond to executions of $\alpha$ in states where it is the case that
\begin{equation*}
\neg\phi \to [\alpha]\sneg\phi
\end{equation*}
\end{example}

\begin{example}\label{exam: inc removal}
More generally, state transitions (programs) on Belnap--Dunn states may be seen as arbitrary modifications of states. Program $\alpha_1$ is locally equivalent to `marking $ \phi$ as true' and $\alpha_2$ to `marking $\psi$ as false' if 
\begin{equation*} [\alpha_1]\phi \land [\alpha_2]\sneg\psi
\end{equation*}
holds in the given state and similarly for $\beta$-equivalence. More interestingly, the formula
\begin{equation*}
(\phi \land \sneg \phi) \land \langle \alpha^{\ast}\rangle \neg (\phi \land \sneg \phi)
\end{equation*}
says that there is inconsistent information about $\phi$ in the present state, but the inconsistency is removed after some finite number of executions of $\alpha$. In other words, $\alpha$ is a $\phi$-inconsistency-removing modification. 

Again, we may see the above formulas as defining the respective programs to be counterparts of specific modifications of states. %On this view, it is natural to focus on models where the formulas are valid.
\end{example}

\section{$\mathbf{BPDL}$}\label{sec: BPDL}
The language $\mathcal{L}_{dyn}$ is a variant of the language of $\mathbf{PDL}$, containing two kinds of expressions, namely, programs $P$ and formulas $F$: 
\begin{align*}
P \quad \alpha & ::= a \mid \alpha ; \alpha \mid \alpha \cup \alpha \mid \alpha^{\ast} \mid \phi? \\
F \quad \phi & ::= p \mid \bot \mid \sneg \phi \mid \phi \land \phi \mid \phi \lor \phi \mid \phi \to \phi \mid [\alpha]\phi \mid \langle \alpha \rangle \phi
\end{align*}
($a \in AP$, a countable set of atomic programs, and $p \in AF$) $\neg \phi$, $\top$ and $\phi \leftrightarrow \psi$ are defined as in $\mathcal{L}_{mod}$.

\begin{definition}
A \emph{standard dynamic Odintsov--Wansing model} is a tuple $\mathcal{M} = \langle S, R, V^{+}, V^{-} \rangle$, where $S, V^{+}$ and $ V^{-}$ are as in Odintsov--Wansing models. $\models_\mathcal{M}^{+}$ and $\models_\mathcal{M}^{-}$ are defined as before for $\{ AF, \bot, \sneg, \land, \lor, \to \}$. $R$ is a function from $P$ to binary relations on $S$ such that $R(\alpha; \beta)$ ($R(\alpha \cup \beta)$) is the composition (union) of $R(\alpha)$ and $R(\beta)$; $R(\alpha^{\ast})$ is the reflexive transitive closure $R(\alpha)^{\ast}$ of $R(\alpha)$; and $R(\phi?)$ is the identity relation on $|\phi|^{+}$. Moreover ($R_\alpha$ is short for $R(\alpha)$): 
\begin{enumerate}
\item $x \models^{+} [\alpha] \phi$ iff for all $y$, if $R_\alpha xy$, then $y \models^{+} \phi$
\item $x \models^{-} [\alpha] \phi$ iff there is $y$ such that $R_\alpha xy$ and $y \models^{-} \phi$\\%[-10pt]
\item $x \models^{+} \langle \alpha \rangle \phi$ iff there is $y$ such that $R_\alpha xy$ and $y \models^{+} \phi$
\item $x \models^{-} \langle \alpha \rangle \phi$ iff for all $y$, if $R_\alpha xy$, then $y \models^{-} \phi$
\end{enumerate}
Entailment in $\mathbf{BPDL}$ is defined as $\models^{+}$-preservation in every state of every standard dynamic Odintsov--Wansing model. Validity in $\mathcal{M}$ and (logical) validity $\models \phi$ are defined as usual. ($\mathcal{M} \models \phi$ iff $x \models^{+}_\mathcal{M} \phi$ for all states $x \in S$ of $\mathcal{M}$; $\models \phi$ iff $\mathcal{M} \models \phi$ for every standard dynamic Odintsov--Wansing model $\mathcal{M}$.) In addition to `local' entailment, we define the \emph{global consequence} relation as follows: $X \models^{g} \phi$ iff, for all $\mathcal{M}$, if every $\psi \in X$ is valid in $\mathcal{M}$, then so is $\phi$.

A \emph{non-standard} dynamic Odintsov--Wansing model is defined exactly as a standard model, with one exception: $R(\alpha^{\ast})$ is required to be a superset of $R(\alpha)^{\ast}$ (the converse inclusion is not assumed) such that
\begin{align}
|[\alpha^{\ast}]\phi|^{+} & = |\phi \land [\alpha][\alpha^{\ast}]\phi|^{+} \label{eq: iteration 1}\\
|[\alpha^{\ast}]\phi|^{+} & \supseteq |\phi \land [\alpha^{\ast}](\phi \to [\alpha]\phi)|^{+}
\end{align}
and
\begin{align}
|\langle \alpha^{\ast}\rangle\phi|^{+} & = |\phi \lor \langle \alpha\rangle\langle \alpha^{\ast}\rangle \phi|^{+}\\
|\langle \alpha^{\ast}\rangle\phi|^{+} & \subseteq |\phi \lor \langle\alpha^{\ast}\rangle (\neg \phi \land \langle\alpha\rangle\phi)|^{+}\label{ewq: iteration 4}
\end{align}
\end{definition}

In dynamic Odintsov--Wansing models, $\phi?$ tests whether $\phi$ is true. Hence, test $\phi?$ executes successfully in two cases: if $\phi$ is only true and if $\phi$ is both true and false. However, if a more precise assessment of $\phi$ is needed, one can use $(\phi \land \neg\sneg \phi)?$ and $(\phi \land \sneg \phi)?$.  

\begin{lemma}\label{lem: bk axioms valid}
All the $H(\mathbf{BK})$ axiom schemata of Theorem \ref{thm: bk complete}, with all $\B$ replaced by $[\alpha]$ and all $\D$ replaced by $\langle \alpha \rangle$, are valid in $\mathbf{BPDL}$. Moreover, the set of formulas valid in any (standard or non-standard) model is closed under Modus ponens and the Necessitation rule $\phi / [\alpha] \phi$.
\end{lemma}

\begin{lemma}\label{lem: pdl axioms valid} The following schemata are valid in every (standard or non-standard) model:
\begin{enumerate}
\item $[\alpha \cup \beta]\phi \leftrightarrow \left([\alpha] \phi \land [\beta]\phi\right)$ and $\langle \alpha \cup \beta \rangle \phi \leftrightarrow \left(\langle \alpha \rangle \phi \lor \langle \beta \rangle \phi\right)$
\item $[\alpha;\beta]\phi \leftrightarrow [\alpha][\beta]\phi$ and $\langle \alpha;\beta\rangle\phi \leftrightarrow \langle \alpha \rangle \langle \beta \rangle \phi$
\item $[\psi?]\phi \leftrightarrow \left( \psi \to \phi \right)$ and $\langle \psi? \rangle \phi \leftrightarrow \left( \psi \land \phi \right)$
\item $[\alpha^{\ast}]\phi \leftrightarrow \left( \phi \land [\alpha][\alpha^{\ast}]\phi\right)$ and $\langle \alpha^{\ast}\rangle\phi \leftrightarrow \left( \phi \lor \langle\alpha\rangle\langle\alpha^{\ast}\rangle\phi\right)$
\item $\left( \phi \land [\alpha^{\ast}](\phi \to [\alpha]\phi)\right) \to [\alpha^{\ast}]\phi$ and $\langle\alpha^{\ast}\rangle\phi \to \left(\phi \lor \langle\alpha^{\ast}\rangle(\neg \phi \land \langle\alpha\rangle \phi)\right)$
\end{enumerate}
\end{lemma}
\begin{proof}
The proofs are virtually identical to arguments used in the context of standard $\mathbf{PDL}$ \cite{Harel2000}. As an example, we show that $[\alpha^{\ast}]\phi \to \left( \phi \land [\alpha][\alpha^{\ast}]\phi\right)$ is valid. The validity of $[\alpha^{\ast}]\phi \to \phi$ follows from the fact that $R(\alpha^{\ast})$ is reflexive. Now if $x \not\models^{+} [\alpha][\alpha^{\ast}]\phi$, then there are $y,z$ such that $R(\alpha)xy$, $R(\alpha^{\ast})yz$ and $z \not\models^{+} \phi$. But obviously $R(\alpha^{\ast})xz$, so $x \not\models^{+} [\alpha^{\ast}]\phi$.
\end{proof}

It is plain that compactness fails for $\mathbf{BPDL}$ for the same reason as for $\mathbf{PDL}$ \cite[181]{Harel2000}. Every finite subset of \[ M = \{\langle \alpha^{\ast} \rangle \phi \} \cup \{ \neg \phi \} \cup \{ \neg \langle \alpha^{n} \rangle \phi \mid n \in \omega \}\] is satisfiable, but $M$ itself is not ($\alpha^{n} = \underbrace{\alpha; \ldots ; \alpha}_{n\text{ times}}$).

Examples \ref{exam: default} -- \ref{exam: inc removal} suggest that some $\mathcal{L}_{dyn}$-formulas can be seen as definitions of specific features of programs ($\alpha$ represents a default rule, $\alpha$ removes inconsistency in the information about a specific formula, etc.). Global consequence is a natural notion here. If $X$ is a set of such definitions, then $X \models^{g} \phi$ iff $\phi$ is valid in every model that respects the definitions `globally'. In other words, $\phi$ is a consequence of the assumption that the definitions in $X$ are satisfied in every possible state. Similarly as in the case of $\mathbf{PDL}$ (see \cite[209]{Harel2000}, global consequence for finite $X$ corresponds to validity of specific formulas.

\begin{proposition}
Let $\{a_1, \ldots, a_n \}$ be the set of all atomic programs appearing in some formula in (finite) $X$ or in $\phi$. Then \[X \models^{g} \phi \iff \models [(a_1 \cup \ldots \cup a_n)^{\ast}]\bigwedge X \to \phi \]
\end{proposition}
\begin{proof}
The right-to-left implication is trivial. The converse implication is established as follows.  If $[(a_1 \cup \ldots \cup a_n)^{\ast}]\bigwedge X \to \phi$ is not valid (the antecedent of this implication is abbreviated as $X^{\ast}$), then there is a state $x$ of a model $\mathcal{M}$ such that $x \models^{+} X^{\ast} \land \neg \phi$. Define $\mathcal{M}_x$ by setting $S_x = \{ y \mid \langle x,y \rangle \in R((a_1 \cup \ldots \cup a_n)^{\ast}) \}$ and taking $R_x, V^{+}_x$ and $V^{-}_x$ to be restrictions of the original $R, V^{+}, V^{-}$ to $S_x$.  It is plain that $\bigwedge X$ is valid in $\mathcal{M}_x$, but $\phi$ is not (the key fact, easily established by induction on the complexity of $\alpha$, is that if every atomic program appearing in $\alpha$ is in $\{a_1 \cup \ldots \cup a_n \}$, then $R(\alpha)zz'$ only if $R_x(\alpha)zz'$, for all $z,z' \in S_x$). Hence, $X \not\models^{g} \phi$. 
\end{proof}

\section{Decidability}\label{sec: decidability}
In this section we establish decidability of the satisfiability problem of $\mathcal{L}_{dyn}$ formulas in (standard and non-standard) dynamic Odintsov--Wansing models. We modify the standard technique using filtration trough the Fischer--Ladner closure of a formula. Our definition of the Fisher--Ladner closure is a simplified version of the definition used in \cite{Harel2000}.
\begin{definition}
The \emph{Fisher-Ladner closure} of $\phi$, $FL(\phi)$, is the smallest set of formulas such that
	\begin{itemize}
	\item $\phi \in FL(\phi)$ and $FL(\phi)$ is closed under subformulas;
	\item if $[\psi?]\chi \in FL(\phi)$, then $\psi \in FL(\phi)$;
	\item if $[\alpha \cup \beta]\chi \in FL(\phi)$, then $[\alpha]\chi \in FL(\phi)$ and $[\beta]\chi \in FL(\phi)$;
	\item if $[\alpha; \beta]\chi \in FL(\phi)$, then $[\alpha][\beta]\chi \in FL(\phi)$;
	\item if $[\alpha^{\ast}]\chi \in FL(\phi)$, then $[\alpha][\alpha^{\ast}]\chi \in FL(\phi)$;
	\item variants of the above conditions with all `$[\cdot]$' replaced by `$\langle \cdot \rangle$'.
	\end{itemize}
\end{definition} 

\begin{lemma}
For all $\phi$, $FL(\phi)$ is finite.
\end{lemma}
\begin{proof}
Standard argument, see \cite{Gold92}.
\end{proof}

\begin{definition}
Let $T$ be a set of formulas and $\mathcal{M}$ a (standard or non-standard) model with $x,y \in S$. Let $x \equiv_T y$ iff, for all $\phi \in T$,
\begin{align*}
x \models^{+}_\mathcal{M} \phi &\iff y \models^{+}_\mathcal{M} \phi\\
x \models^{-}_\mathcal{M} \phi &\iff y \models^{-}_\mathcal{M} \phi.
\end{align*}
Let $[x]_T = \{ y \mid x \equiv_T y \}$. The \emph{filtration of $\mathcal{M}$ trough $T$} is $\mathcal{M}_T = \langle S_T, R_T, V^{+}_T, V^{-}_T \rangle$, where
\begin{enumerate}
\item $S_T = \{ [x]_T \mid x \in S \}$;
\item $R_T(a) = \{ \langle [x]_T, [y]_T \rangle \mid R_a xy \}$ for all $a \in AP$;
\item $V_T^{+}(p) = \{ [x]_T \mid x \in V^{+}(p) \} $;
\item $V_T^{-}(p) = \{ [x]_T \mid x \in V^{-}(p) \} $.
\end{enumerate} 
Relations $\models^{+}_{\mathcal{M}_T}$, $\models^{-}_{\mathcal{M}_T}$ and $R_T(\alpha)$ for complex $\alpha$ are defined as in standard models.
\end{definition}

\noindent It is plain that $\mathcal{M}_T$ is a standard model. We write $[x]$ instead of $[x]_T$ if $T$ is clear from the context.

\begin{lemma}\label{lem: filtration cardinality}
For all $\mathcal{M}_T$, $|S_T| \leq 4^{|T|}$.
\end{lemma}
\begin{proof}
There are four possible truth values of each member of $T$.
\end{proof}

\begin{lemma}[Filtration Lemma]\label{lem: filtration} Let $\mathcal{M}$ be a (standard or non-standard) model and $\phi$ a formula.
\begin{enumerate}
\item If $[\alpha]\psi \in FL(\phi)$ or $\langle \alpha \rangle \psi \in FL(\phi)$, then $R(\alpha) xy$ only if $R_{FL(\phi)}(\alpha)[x][y]$;
\item If $[\alpha]\psi \in FL(\phi)$, then $R_{FL(\phi)}(\alpha)[x][y]$ and $x \models^{+} [\alpha]\psi$ only if $y \models^{+} \psi$;
\item If $\langle\alpha\rangle\psi \in FL(\phi)$, then $R_{FL(\phi)}(\alpha)[x][y]$ and $y \models^{+}\psi$ only if $x \models^{+} \langle \alpha\rangle \psi$;
\item If $\langle\alpha\rangle\psi \in FL(\phi)$, then $R_{FL(\phi)}(\alpha)[x][y]$ and $x \models^{-} \langle\alpha\rangle\psi$ only if $y \models^{-} \psi$;
\item If $[\alpha]\psi \in FL(\phi)$, then $R_{FL(\phi)}(\alpha)[x][y]$ and $y \models^{-} \psi$ only if $x \models^{-} [\alpha] \psi$;
\item If $\psi \in FL(\phi)$, then $x \models^{+} \psi$ iff $[x] \models^{+} \psi$;
\item If $\psi \in FL(\phi)$, then $x \models^{-} \psi$ iff $[x] \models^{-} \psi$.
\end{enumerate}
\end{lemma}
\begin{proof}
A simple but tedious variation of the standard proof using simultaneous induction on the subexpression relation \cite{Gold92,Harel2000}. Details of some of the steps are given in Appendix \ref{sec: appendix}.
\end{proof}

\begin{theorem}
The satisfiability problem for $\mathbf{BPDL}$ is decidable.
\end{theorem}
\begin{proof}
Standard argument. If $\phi$ is satisfiable in some $\mathcal{M}$, then, by Lemmas \ref{lem: filtration cardinality} and \ref{lem: filtration}(vi), $\phi$ is satisfiable in a standard model of size at most $4^{k}$ where $k = |FL(\phi)|$. There is a finite number of such models, so a naive satisfiability algorithm is to determine $k = |FL(\phi)|$ and check all models of size $4^{k}$. 
\end{proof}

\section{Completeness}\label{sec: completeness}
The axiom system $H(\mathbf{BPDL})$ results from $H(\mathbf{BK})$ by replacing all `$\B$' by `$[\alpha]$' and all `$\D$' by `$\langle \alpha \rangle$' and adding the schemata explicitly stated in Lemma \ref{lem: pdl axioms valid}. (See Appendix \ref{sec: appendix b}.) The notion of a maximal $H(\mathbf{BPDL})$-consistent set (m.c. set) of formulas is defined as usual ($X$ is consistent iff $\neg \bigwedge X'$ is not provable for all finite $X' \subseteq X$; $X$ is m.c. iff $X$ is consistent all $X' \supset X$ are inconsistent). Hence, m.c. sets have all the usual properties. 

\begin{definition}
The canonical model $\mathcal{M}_c = \langle S_c, R_c, V^{+}_c, V^{-}_c \rangle$ is a quadruple such that
\begin{enumerate}
 \item $S_c$ is the set of all m.c. sets;
 \item $R_c(\alpha)XY$ iff for all $[\alpha]\phi \in X$, $\phi \in Y$ (iff for all $\phi \in Y$, $\langle \alpha \rangle \phi \in X$);
 \item $V^{+}_c(p) = \{ X \mid p \in X \}$; 
  \item $V^{-}_c(p) = \{ X \mid \sneg p \in X \}$;
 \end{enumerate}
$|\phi|^{+}_c = \{ X \mid \phi \in X \}$ and $|\phi|^{-}_c = \{ X \mid \sneg \phi \in X \}$.
\end{definition}

\begin{lemma}
$|\phi|^{+}_c$ and $|\phi|^{-}_c$ behave like $|\phi|^{+}$ and $|\phi|^{-}$ (in standard and non-standard models), respectively:
\begin{itemize}
\item $X \in |p|^{+}_c$ iff $X \in V^{+}_c(p)$; $X \in |p|^{-}_c$ iff $ X \in V^{-}_c(p)$;
\item $|\bot|^{+}_c = \emptyset$; $|\bot|^{-}_c = S_c$;
\item $|\sneg \phi|^{+}_c = |\phi|^{-}_c$; $|\sneg \phi|^{-}_c = |\phi |^{+}_c$;
\item $|\phi \land \psi |^{+}_c = |\phi|^{+}_c \cap |\psi|^{+}_c$; $|\phi \land \psi|^{-}_c = |\phi|^{-}_c \cup |\psi|^{-}_c$; 
\item $|\phi \lor \psi |^{+}_c = |\phi|^{+}_c \cup |\psi|^{+}_c$; $|\phi \land \psi|^{-}_c = |\phi|^{-}_c \cap |\psi|^{-}_c$;
\item $|\phi \to \psi|^{+}_c = (S_c - |\phi|^{+}_c) \cup |\psi|^{+}_c$; $|\phi \to \psi|^{-}_c = |\phi|^{+}_c \cap |\psi|^{-}_c$;
\item $|[\alpha]\phi|^{+}_c = \{ X \mid (\forall Y)(\text{if }R_c(\alpha)XY,\text{then } Y \in |\phi|^{+}_c ) \}$;\\ $|[\alpha]\phi|^{-}_c = \{X \mid (\exists Y)(R_c(\alpha)XY \text{and } Y \in |\phi|^{-}_c)\}$;
\item $|\langle \alpha \rangle\phi|^{+}_c = \{X \mid (\exists Y)(R_c(\alpha)XY \text{and } Y \in |\phi|^{+}_c)\}$;\\ $|\langle \alpha \rangle \phi|^{-}_c = \{ X \mid (\forall Y)(\text{if }R_c(\alpha)XY,\text{then } Y \in |\phi|^{-}_c ) \}$.
\end{itemize}
\end{lemma}
\begin{proof}
Standard inductive argument, we state only three cases explicitly. Firstly, $|\sneg\phi|^{-}_c = \{ X \mid \sneg\sneg \phi \in X \}$ and, as $\phi \leftrightarrow \sneg\sneg\phi$ is an axiom, this set is identical to $\{ X \mid \phi \in X \}$, i.e., to $|\phi|^{+}_c$. 

Secondly, $|[\alpha]\phi|^{+} = \{ X \mid [\alpha]\phi \in X \}$. We have to show that $[\alpha]\phi \in X$ iff for all $Y$, $R_c(\alpha)XY$ only if $Y \in |\phi|^{+}$. The left-to-right implication is trivial. The right-to-left implication is established by the following standard argument. Assume that $\neg [\alpha] \phi \in X$. We want to show that there is $Y$ such that $R_c(\alpha)XY$ and $\neg \phi \in Y$. We claim that the set
\begin{equation}
M = \{ \neg  \phi \} \cup \{ \psi \mid [\alpha]\psi  \in X \}
\end{equation}
is consistent. (We denote $\{ \psi \mid [\alpha]\psi  \in X \}$ as $X^{-\alpha}$.) To see this, take an arbitrary finite $\Psi = \{\psi_1, \ldots, \psi_m\} \subseteq X^{-\alpha}$. It is plain that $\langle \alpha \rangle \neg \phi \land [\alpha]\psi_1  \land \ldots \land [\alpha]\psi_m \in X$. Hence, by the $\mathbf{K}$-style properties of $[\alpha]$ and $\langle\alpha\rangle$, $\langle\alpha\rangle ( \neg \phi \land \psi_1 \land \ldots \land \psi_m) \in X$. By the Necessitation rule, $\neg (\neg \phi \land \psi_1 \land \ldots \land \psi_m)$ is not provable, so $\{ \neg \phi, \psi_1, \ldots, \psi_m\}$ is consistent. But $\Psi$ was chosen as an arbitrary finite subset of $X^{-\alpha}$. Consequently, $\{ \neg \phi \} \cup \Psi'$ for every finite $\Psi' \subseteq X^{-\alpha}$ can be shown to be consistent in this way. Hence, $M$ itself is consistent. By the Lindenbaum Lemma, $M$ can be extended to a m.c. $Y$ and it is plain that $R_c(\alpha)XY$ and $\neg \phi \in Y$.

Thirdly, $|[\alpha]\phi|^{-}_c = \{ X \mid \sneg [\alpha]\phi \in X \}$, a set identical to $\{ X \mid \langle \alpha \rangle \sneg \phi \in X \}$ as $\sneg [\alpha]\phi \leftrightarrow \langle\alpha\rangle\sneg\phi$ is an axiom. A straightforward adaptation of the argument given by \cite[p.~206]{Harel2000} shows that this set is identical to $\{ X \mid (\exists Y)(R_c(\alpha)XY \text{and } Y \in |\phi|^{-}_c) \}$. 
\end{proof}

\begin{lemma}\label{lem: canonical model lemma}
$\mathcal{M}_c$ is a non-standard model.
\end{lemma}
\begin{proof}
We need to be establish that $R$ satisfies the conditions required by the definition of a non-standard model. The argument for $\cup$, $;$, $?$ and the iteration equations \eqref{eq: iteration 1} -- \eqref{ewq: iteration 4} is virtually identical to that given by \cite[p.206--8]{Harel2000}. To show that $R_c(\alpha)^{\ast} \subseteq R_c(\alpha^{\ast})$, assume that $\langle X,Y \rangle \in R_c(\alpha)^{\ast}$ but $\langle X,Y \rangle \not\in R_c(\alpha^{\ast})$. Hence, there is $\phi$ such that $[\alpha^{\ast}]\phi \in X$ but $\phi \not\in Y$. However, $\langle X,Y \rangle \in R_c(\alpha)^{\ast}$ implies that either $X = Y$ or else there are $Z_0, \ldots, Z_m$ such that $Z_0 = X$, $Z_m = Y$ and $\langle Z_k, Z_{k+1} \rangle \in R_c(\alpha)$ for $1 \leq k < m$. In the former case, $\phi \in Y$ by the axiom $[\alpha^{\ast}]\phi \leftrightarrow \left( \phi \land [\alpha][\alpha^{\ast}]\phi\right)$, a contradiction. In the latter case, $[\alpha^{\ast}]\phi \in Z_k$ entails $[\alpha][\alpha^{\ast}]\phi \in Z_{k+1}$ by the same axiom for all $1 \leq k < n$ and, hence, $[\alpha^{\ast}]\phi \in Y$. Hence, $\phi \in Y$, a contradiction.  
\end{proof}

Since $\mathbf{BPDL}$ is not compact, it cannot enjoy a strongly complete axiomatisation (as $\mathbf{BK}$ does). However, weak completeness is another story.

\begin{theorem}
$\phi$ is provable in $H(\mathbf{BPDL})$ iff $\phi$ is valid in $\mathbf{BPDL}$.
\end{theorem}
\begin{proof}
Soundness follows from Lemmas \ref{lem: bk axioms valid} and \ref{lem: pdl axioms valid}. Completeness follows from Lemmas \ref{lem: filtration} and \ref{lem: canonical model lemma}. If $\phi$ is not provable, then $X \in |\neg \phi |^{+}_c $ for some m.c. set $X$. By the Filtration Lemma, $(\mathcal{M}_c)_{FL(\neg\phi)}$ is a standard model such that $[X] \in |\neg\phi|^{+}_{FL(\neg\phi)}$.
\end{proof}

\section{Conclusion}
This article introduced $\mathbf{BPDL}$, a combination of propositional dynamic logic $\mathbf{PDL}$ with the four-valued Belnapian modal logic $\mathbf{BK}$. The logic is expected to be useful in formalising reasoning about the properties of algorithmic transformations of possibly incomplete and inconsistent database-like bodies of information. We modified the standard proofs based on filtration and the canonical-model technique and, as the main technical results of the article, established decidability of $\mathbf{BPDL}$ and provided it with a sound and weakly complete axiomatisation. The main message here is that the standard techniques are easily adapted to the four-valued setting. 

The number one topic for future research is the complexity of the satisfiability problem for $\mathbf{BPDL}$. The problem is $\mathit{EXPTIME}$-complete for $\mathbf{PDL}$ and it will be interesting to see whether the situation gets worse in the case of $\mathbf{BPDL}$. Our strategy of tackling the problem will be, as for the results already achieved, to try to adapt the proof technique used in the case of $\mathbf{PDL}$ to the four-valued setting. We shall also investigate Belnapian versions of some extensions of $\mathbf{PDL}$. The obvious choice is the first-order dynamic logic $\mathbf{DL}$, but also concurrent $\mathbf{PDL}$ modelling parallel execution of programs. %As in the test case of $\mathbf{BPDL}$, we will look at the possibilities of adapting the existing techniques to the four-valued setting. 
Last but not least, a more thorough examination of possible applications of $\mathbf{BPDL}$ will be an interesting enterprise. %Generally speaking, the present article is thought of as a first step in the study of a wider range non-classical dynamic logics.

\appendix
\section{Proof of the Filtration Lemma}\label{sec: appendix}
The proof is a variation of the standard proof using simultaneous induction on the subexpression (subformula or subprogram) relation \cite{Gold92,Harel2000}. In proving the claim of any item (i)--(vii) for any special case of $\alpha$ or $\psi$, we assume that all the items hold for all subexpressions of $\alpha$ and $\psi$. Only some steps of the proof are explicitly stated here (and, perhaps, in more detail than an expert reader needs).

\subsection{$\mathcal{M}$ is a standard model}
\emph{(i), $\alpha = \beta^{\ast}$.} If $R(\beta^{\ast})xy$, then, since $R(\beta^{\ast})$ is the reflexive transitive closure of $R(\beta)$, there are $z_0, \ldots, z_n$ such that $z_0 = x$, $z_n = y$ and either $n = 0$ or else $R(\beta)z_iz_{i+1}$ for $0 \leq i < n$. If $n = 0$, then $R_{FL(\phi)}(\beta^{\ast})[z_0][z_n]$ by the definition of $R_{FL(\phi)}(\beta^{\ast})$. Assume $n > 0$. If $[\beta^{\ast}]\psi$ ($\langle \beta^{\ast}\rangle$) is in $FL(\phi)$, then so is $[\beta][\beta^{\ast}]\psi$ ($\langle \beta \rangle \langle \beta^{\ast} \rangle \psi$). $\beta$ is a subexpression of $\beta^{\ast}$, so, in both cases, we may apply the induction hypothesis (IH): $R(\beta)z_iz_{i+1}$ implies $R_{FL(\phi)}(\beta)[z_i][z_{i+1}]$ for $0 \leq i < n$. Hence, $R_{FL(\phi)}(\beta^{\ast})[z_0][z_n]$ by the definition of $R_{FL(\phi)}$.

\emph{(i), $\alpha = \chi?$.} If $R(\chi?)xy$, then $x = y$ and $x \models^{+} \chi$. If $[\chi?]\psi$ ($\langle \chi?\rangle \psi$) is in $FL(\phi)$, then so is $\chi$. $\chi$ is a subexpression of $[\chi?]\psi$ ($\langle \chi? \rangle \psi$) and a formula, so we may apply the IH of (vi): $x \models^{+} \chi$ entails $[x] \models^{+} \chi$. But $[x] = [y]$ and, hence, $R_{FL(\phi)}(\chi?)[x][y]$ by the definition of $R_{FL(\phi)}$.

\emph{(ii), $\alpha = \beta^{\ast}$.} If $R_{FL(\phi)}(\beta^{\ast})[x][y]$, then there are $z_0, \ldots, z_n$ such that $[z_0] = [x]$, $[z_n] = [y]$ and either $n = 0$ or else $R_{FL(\phi)}(\beta)[z_i][z_{i+1}]$ for $0 \leq i < n$. If $n = 0$, then $x \models^{+} [\beta^{\ast}]\psi$ entails $y \models^{+} \psi$ by the assumption $[\beta^{\ast}]\psi \in FL(\phi)$ and Lemma \ref{lem: pdl axioms valid}(iv). Assume $n > 0$. We prove that 
\begin{equation}\label{eq: claim1}
x \models^{+} [\beta^{\ast}]\psi \Longrightarrow z_k \models^{+} [\beta^{\ast}]\psi \qquad (0 \leq k \leq n)
\end{equation}
by induction on $k$. If $k = 0$, then the claim follows from the assumption $[\beta^{\ast}]\psi \in FL(\phi)$. Assume that the claim holds for $k = l$. We prove that it holds for $k = l+1$ as well. The assumption is that $x \models^{+} [\beta^{\ast}]\psi$ entails $z_l \models^{+} [\beta^{\ast}] \psi$. By Lemma \ref{lem: pdl axioms valid}(iv), $z_l \models [\beta][\beta^{\ast}]\psi$. $\beta$ is a subexpression of $\beta^{\ast}$ and $[\beta][\beta^{\ast}]\psi \in FL(\phi)$, so we may use IH of item (ii) of the Filtration Lemma: $R_{FL(\phi)}(\beta)[z_k][z_{l+1}]$ entails $z_{l+1} \models^{+} [\beta^{\ast}]\psi$. This proves \eqref{eq: claim1}. Now $x \models^{+} [\beta^{\ast}]\psi$ entails $z_n \models [\beta^{\ast}]\psi$ by \eqref{eq: claim1} and $z_n \models [\beta^{\ast}]\psi$ entails $z_n \models^{+} \psi$ by Lemma \ref{lem: pdl axioms valid}(iv). But $[z_n] = [y]$ and $\psi \in FL(\phi)$, so $y \models^{+} \psi$.

\emph{(iv), $\alpha = \beta^{\ast}$.} If $R_{FL(\phi)}(\beta^{\ast})[x][y]$, then there are $z_0, \ldots, z_n$ such that $[z_0] = [x]$, $[z_n] = [y]$ and either $n = 0$ or else $R_{FL(\phi)}(\beta)[z_i][z_{i+1}]$ for $0 \leq i < n$.  If $n = 0$, then $x \models^{-} \langle\beta^{\ast}\rangle\psi$ entails $y \models^{-} \langle \beta^{\ast} \rangle \psi$ by the assumption $\langle \beta^{\ast} \rangle \psi \in FL(\phi)$. Hence, $y \models^{+} [\beta^{\ast}] \sneg \psi$ by Lemma \ref{lem: bk axioms valid}. By Lemma \ref{lem: pdl axioms valid}(iv), $y \models^{+} \sneg \psi$. Hence, $y \models^{-} \psi$. Next, assume that $n > 0$. We prove that 
\begin{equation}\label{eq: claim2}
x \models^{-} \langle\beta^{\ast}\rangle\psi \Longrightarrow z_k \models^{-} \langle\beta^{\ast}\rangle\psi \qquad (0 \leq k \leq n)
\end{equation}
by induction on $k$. If $k = 0$,  then the claim follows from the assumption $\langle\beta^{\ast}\rangle\psi \in FL(\phi)$. Assume that the claim holds for $k = l$. We prove that it holds for $k = l+1$ as well. The assumption is that $x \models^{-} \langle\beta^{\ast}\rangle\psi$ entails $z_l \models^{-} \langle\beta^{\ast}\rangle \psi$. By Lemmas \ref{lem: bk axioms valid} and \ref{lem: pdl axioms valid}(iv), $z_l \models^{-} \langle \beta \rangle \langle \beta^{\ast} \rangle \psi$. $\beta$ is a subexpression of $\beta^{\ast}$, so we may use the IH to infer $z_{l+1} \models^{-} \langle \beta^{\ast} \rangle \psi$. This proves \eqref{eq: claim2}. Assume that $x \models^{-} \langle \beta^{\ast} \rangle \psi$. By \eqref{eq: claim2}, $z_n \models^{-} \beta^{\ast} \rangle \psi$. By the assumption that $\langle\beta^{\ast}\rangle\psi \in FL(\phi)$, $y \models^{-}  \beta^{\ast} \rangle \psi$. By Lemmas \ref{lem: bk axioms valid} and \ref{lem: pdl axioms valid}(iv), $y \models^{-} \psi$.

\emph{(v), $\alpha = \beta^{\ast}$.} If $R_{FL(\phi)}(\beta^{\ast})[x][y]$, then there are $z_0, \ldots, z_n$ such that $[z_0] = [x]$, $[z_n] = [y]$ and either $n = 0$ or else $R_{FL(\phi)}(\beta)[z_i][z_{i+1}]$ for $0 \leq i < n$. If $n = 0$, then the reasoning is similar as in the above cases. Hence, assume that $n > 0$. We prove that 
\begin{equation}\label{eq: claim3}
z_n \models^{-} \psi \Longrightarrow z_{n-k} \models^{-} [\beta^{\ast}]\psi \qquad (0 \leq k \leq n)
\end{equation}
by induction on $k$. The case $k = 0$is trivial. Assume that the claim holds for $k = l$. We prove that it holds for $k = l+1$ as well. The assumption is that $z_n \models^{-} \psi$ entails $z_{n-l} \models^{-} [\beta^{\ast}]\psi$. There are two possibilities. Either (a) $z_{n-l} \models^{-} \psi$ or (b) $z_{n-l} \models^{-} [\beta][\beta^{\ast}]\psi$. If (a), then $z_{n-(l+1)} \models^{-} [\beta]\psi$ by IH and $z_{n-(l+1)} \models^{-} [\beta^{\ast}] \psi$ by $R_\beta \subseteq (R_\beta)^{\ast}$. If (b), then IH entails that $z_{n - (l+1)} \models^{-} [\beta][\beta][\beta^{\ast}] \psi$. By $R_\beta \subseteq (R_\beta)^{\ast}$, $z_{n - (l+1)} \models^{-} [\beta^{\ast}]\psi$. This proves \eqref{eq: claim3}. Now $y \models^{-} \psi$ only if $z_n \models^{-} \psi$ ($\psi \in FL(\phi)$) only if $z_0 \models^{-} [\beta^{\ast}]\psi$ \eqref{eq: claim3} only if $x \models^{-} [\beta^{\ast}]\psi$ ($[\beta^{\ast}]\psi \in FL(\phi)$). 

\emph{(vi), $\psi = \sneg \chi$.} $x \models^{+} \sneg \chi$ iff $x \models^{-} \chi$. $\chi$ is a subexpression of $\sneg \chi$, so we may use IH of (vii): and infer $x \models^{-} \chi$ iff $[x] \models^{-} \chi$ iff $[x] \models^{+} \sneg \chi$ (by the definition of $\models^{+}_{\mathcal{M}_{FL(\phi)}}$. In fact, this case requires to introduce item (vii) into the Filtration Lemma. 

\emph{(vi), $\psi = [\alpha]\chi$.} Assume $[\alpha]\chi \in FL(\phi)$. Then $\chi \in FL(\phi)$. To prove the left-to-right implication, assume that $x \models^{+} [\alpha]\chi$ and $R_{FL(\phi)}(\alpha)[x][y]$. $\alpha$ is a subexpression of $[\alpha]\chi$, so we may use IH of (ii) to infer $y \models^{+} \chi$. By IH, $[y] \models^{+} \chi$. To prove the right-to-left implication, assume that $[x] \models^{+} [\alpha]\chi$ and $R_\alpha xy$. IH of (i) implies that $R_{FL(\phi)}(\alpha)[x][y]$. Consequently, $[y] \models^{+} \chi$ and, by IH, $y \models^{+} \chi$.

\emph{(vi), $\psi = \langle \alpha \rangle \chi$.} We prove only the right-to-left implication. If $[x] \models^{+} \langle \alpha \rangle \chi$, then there is $y$ such that $R_{FL(\phi)}(\alpha)[x][y]$ and $[y] \models^{+} \chi$. By IH, $y \models^{+} \chi$. By IH of (iii), $x \models^{+} \langle \alpha \rangle \chi$. This was the reason we had to include item (iii) of the Lemma.

\emph{(vii), $\psi = p$.} Let $p \in FL(\phi)$. The left-to-right implication is trivial. To prove the converse, assume that $[x] \models^{-} p$. This means that there is $x' \equiv x$ such that $x' \models^{-} p$. By the definition of filtration, $x \models^{-} p$ as well. Note that to prove this implication it was necessary to define $\equiv$ in terms of both $\models^{+}$ and $\models^{-}$.

\emph{(vii), $\psi = [\alpha]\chi$.} We prove only the right-to-left implication. If $[x] \models^{-} [\alpha]\chi$, then there is $y$ such that $R_{FL(\phi)}(\alpha)[x][y]$ and $[y] \models^{-} \chi$. By IH of (v), $y \models^{-} [\alpha]\chi$. This case required to introduce item (v).

\emph{(vii), $\psi = \langle \alpha \rangle \chi$.} We prove only the left-to-right implication. Assume $x \models^{-} \langle \alpha \rangle \chi$ and $R_{FL(\phi)}(\alpha)[x][y]$. By IH of (iv), $y \models^{-} \chi$. By IH, $[y] \models^{-} \chi$. Hence, $[x] \models^{-} \langle \alpha \rangle \chi$. This case required to introduce item (iv) of the Lemma.

Proofs of other cases are similar or standard. $\qed$

\subsection{$\mathcal{M}$ is a non-standard model}
As in the standard proof for this case, the only claim where the assumption $R_{\alpha^{\ast}} = (R_\alpha)^{\ast}$ was used is (i), $\alpha = \beta^{\ast}$. Hence, we have to prove that if $[\beta^{\ast}]\psi \in FL(\phi)$ or $\langle \beta^{\ast} \rangle \psi \in FL(\phi)$, then $R(\beta^{\ast}) xy$ only if $R_{FL(\phi)}(\beta^{\ast})[x][y]$. Our argument is very close to the one given in \cite[sec. 6.3]{Harel2000}.

Assume that $\langle x,y \rangle \in R(\beta^{\ast})$. We want to show that $\langle [x],[y] \rangle \in R_{FL(\phi)}(\beta^{\ast})$, or equivalently that $y \in E$, where
\begin{align*}
E & = \{ z \mid \langle [x],[z]\rangle \in  R_{FL(\phi)}(\beta^{\ast}) \}
\end{align*}
Recall that $[\cdot]$ is given by some specific finite set$FL(\phi)$ of formulas. For any $[z]_{FL(\phi)}$, define $X_{[z]}$ to be the smallest set of formulas such that, for all $\chi \in FL(\phi)$:
\begin{itemize}
\item If $z \models^{+} \chi$, then $\chi \in X_{[z]}$;
\item If $z \not\models^{+} \chi$, then $\neg \chi\in X_{[z]}$.
\end{itemize}
(Note that we are using `$\neg$' not `$\sneg$'.) Obviously, $X_{[z]}$ is finite for all $z$. Define \[ \psi_{[z]} = \bigwedge X_{[z]} \] It is not hard to show that, for all $w \in S$,
\begin{equation}\label{eq: claim4}
w \models^{+} \psi_{[z]} \iff w \equiv z
\end{equation}
(For instance, assume that $w \not\equiv z$ because there is $\theta \in FL(\phi)$ such that $w \models^{-} \theta$ and $z \not\models^{-} \theta$. But then $z \not\models^{+} \sneg \theta$ and, consequently, $\neg\sneg\theta \in X_{[z]}$. But then $w \not\models^{+} \psi_{[z]}$ because $w \models^{+} \sneg \theta$.) Now define \[ \psi_E = \bigvee_{z \in E} \psi_{[z]}\] It is not hard to show that $\psi_E$ defines $E$, i.e., for all $w \in S$
\begin{equation}\label{eq: claim5}
w \in E \iff w \models^{+} \psi_E
\end{equation}
(For instance, assume that $w \in E$ but $w \not\models^{+} \psi_E$. Then $w \not\models^{+} \psi_{[z]}$ for all $z \in E$. In particular, $w \not\models^{+} \psi_{[w]}$. \eqref{eq: claim4} entails that this is impossible.)

It is easy to show that $E$ is closed under $R_\beta$, i.e., for all $z, z'$,
\begin{equation}\label{eq: claim6}
z \in E \:\: \& \:\: R_\beta zz' \quad \Longrightarrow \quad z' \in E
\end{equation}
($\beta$ is a subexpression of $\beta^{\ast}$, so $R_\beta zz'$ entails $R_{FL(\phi)}(\beta)[z][z']$ by IH. By the definition of $E$, $z \in E$ means that $R_{FL(\phi)}(\beta^{\ast})[x][z]$. Consequently, $R_{FL(\phi)}(\beta^{\ast})[x][z']$. In other words, $z' \in E$.) \eqref{eq: claim6} means that $\psi_E \to [\beta] \psi_E$ is valid in $\mathcal{M}$. By Lemma \ref{lem: bk axioms valid}(Nec. rule), so is $[\beta^{\ast}] \left(\psi_E \to [\beta] \psi_E \right)$. By Lemma \ref{lem: pdl axioms valid}(v), the induction axiom $\left( \phi \land [\alpha^{\ast}](\phi \to [\alpha]\phi)\right) \to [\alpha^{\ast}]\phi$ is also valid in $\mathcal{M}$. 

It is also easy to show that $x \in E$ ($R(\beta^{\ast})$ is a superset of the reflexive transitive closure of $R(\beta)$, so it contains the identity relation of $S_{FL(\phi)}$.) Hence, $x \models^{+} \psi_E \land [\beta^{\ast}] \left(\psi_E \to [\beta] \psi_E \right)$. By the validity of the induction axiom in $\mathcal{M}$, $x \models^{+} [\beta^{\ast}]\psi_E$. Hence, if $R(\beta^{\ast})xy$, then $y \models^{+} \psi_E$. By \eqref{eq: claim5}, $y \in E$. $\qed$

\section{The axiom system $H(\mathbf{BPDL})$}\label{sec: appendix b}
\begin{enumerate}
\item Axioms of classical propositional logic in the language $\{ AF, \bot, \to, \land, \lor \}$ and Modus ponens;
\item Strong negation axioms:
	\begin{align*}
	\sneg \sneg \phi  & \leftrightarrow \phi,\\
	\sneg (\phi \land \psi) & \leftrightarrow (\sneg \phi \lor \sneg \psi),\\
	\sneg (\phi \lor \psi) & \leftrightarrow (\sneg \phi \land \sneg \psi),\\
	\sneg (\phi \to \psi) & \leftrightarrow (\phi \land \sneg \psi),\\
	\top &\leftrightarrow \sneg \bot;
	\end{align*}
\item Modal axiom $[\alpha] (\phi \to \psi) \to ([\alpha] \phi \to [\alpha] \psi)$ and the Necessitation rule $\phi / [\alpha] \phi$;
\item $\mathbf{PDL}$ axiom schemata
	\begin{align*}
[\alpha \cup \beta]\phi \leftrightarrow \left([\alpha] \phi \land [\beta]\phi\right) &\text{ and } \langle \alpha \cup \beta \rangle \phi \leftrightarrow \left(\langle \alpha \rangle \phi \lor \langle \beta \rangle \phi\right),\\
[\alpha;\beta]\phi \leftrightarrow [\alpha][\beta]\phi &\text{ and } \langle \alpha;\beta\rangle\phi \leftrightarrow \langle \alpha \rangle \langle \beta \rangle \phi,\\
[\psi?]\phi \leftrightarrow \left( \psi \to \phi \right) &\text{ and } \langle \psi? \rangle \phi \leftrightarrow \left( \psi \land \phi \right),\\
[\alpha^{\ast}]\phi \leftrightarrow \left( \phi \land [\alpha][\alpha^{\ast}]\phi\right) &\text{ and } \langle \alpha^{\ast}\rangle\phi \leftrightarrow \left( \phi \lor \langle\alpha\rangle\langle\alpha^{\ast}\rangle\phi\right),\\
\left( \phi \land [\alpha^{\ast}](\phi \to [\alpha]\phi)\right) \to [\alpha^{\ast}]\phi &\text{ and } \langle\alpha^{\ast}\rangle\phi \to \left(\phi \lor \langle\alpha^{\ast}\rangle(\neg \phi \land \langle\alpha\rangle \phi)\right);
	\end{align*}
\item Modal interaction principles:
	\begin{align*}
	\neg [\alpha] \phi & \leftrightarrow \langle\alpha\rangle \neg \phi,\\
	\neg \langle\alpha\rangle \phi & \leftrightarrow [\alpha] \neg \phi,\\
	\sneg [\alpha] \phi & \leftrightarrow \langle\alpha\rangle \sneg \phi,\\
	[\alpha] \phi & \leftrightarrow \sneg \langle\alpha\rangle \sneg \phi,\\
	\sneg \langle\alpha\rangle \phi & \leftrightarrow [\alpha] \sneg \phi,\\
	\langle\alpha\rangle \phi & \leftrightarrow \sneg [\alpha] \sneg \phi.
	\end{align*}
\end{enumerate} 
 %%%%%%%%%%%%%%%%%%%%%%%%%%%%%%%%%%%%%%%%

\end{document}